\documentclass{elsarticle}
\usepackage{graphicx}
\usepackage{color}
\usepackage{amsmath}
\usepackage{amssymb}
\usepackage{mathrsfs}         %%% this is used only to generate a curly G, via \mathscr{G}
\usepackage[all]{xy}

\newtheorem{lemma}{Lemma} 
\newtheorem{theorem}{Theorem} 
\newtheorem{corollary}{Corollary} 
\newtheorem{proposition}{Proposition}

\newtheorem{definition}{Definition} 
\newproof{proof}{Proof}

\newcommand{\W}{\textsc{W}}
\newcommand{\WSAT}{\textsc{W[SAT]}}

\newcommand{\FPT}{\textsc{FPT}}
\newtheorem{observation}{Observation}
\begin{document}

\begin{frontmatter}

\title{Relativization makes contradictions harder for Resolution
\tnoteref{CSR}
}
\author{Stefan Dantchev}
\address{
School of Engineering and Computing Sciences, Durham University,\\
  Science Labs, South Road, Durham DH1 3LE, U.K.}
\author{Barnaby Martin}
\address{School of Science and Technology, Middlesex University,\\
The Burroughs, Hendon, London NW4 4BT, U.K.}
\tnotetext[CSR]{Extended abstracts of results in this paper appeared as \cite{Rel-sep} at CSR 2006 and as \cite{CSR2013DantchevMartin} at CSR 2013. Several proofs, especially those omitted from  \cite{Rel-sep}, appear here for the first time.}

\begin{abstract}
We provide a number of simplified and improved separations between pairs of Resolution-with-bounded-conjunction refutation systems, Res$(d)$, as well as their tree-like versions, Res$^*(d)$. The contradictions we use are natural combinatorial principles: the \emph{Least number principle}, LNP$_n$ and an ordered variant thereof, the \emph{Induction principle}, IP$_n$.

LNP$_n$ is known to be easy for Resolution. We prove that its relativization is hard for Resolution, and more generally, the relativization of LNP$_n$ iterated $d$ times provides a separation between Res$(d)$ and Res$(d+1)$. We prove the same result for the iterated relativization of IP$_n$, where the tree-like variant Res$^*(d)$ is considered instead of Res$(d)$.

We go on to provide separations between the parameterized versions of Res$(1)$ and Res$(2)$. Here we are able again to use the relativization of the LNP$_n$, but the classical proof breaks down and we are forced to use an alternative. Finally, we separate the parameterized versions of Res$^*(1)$ and Res$^*(2)$. Here, the relativization of IP$_n$ will not work as it is, and so we make a vectorizing amendment to it in order to address this shortcoming.
\end{abstract}

\end{frontmatter}

\section{Introduction}

We study the power of relativization in \emph{Propositional proof complexity}, \mbox{i.e.} we are
interested in the following question: given a propositional proof system is there a first-order (FO) sentence which is easy but whose relativization is hard (within the system)?
The main motivation for studying relativization comes from a work of Kraj\'{\i}\v{c}ek,
\cite{Combinatorics_of_FO}. He defines a combinatorics of FO structure and a relation of covering
between FO structures and propositional proof systems. The combinatorics contains
all the sentences easy for the proof system. On the other hand, as defined in \cite{Combinatorics_of_FO}, it is closed under relativization. Thus the existence of a sentence, which is easy but whose relativization is hard, for the underlying proof system, would imply that it is impossible
to capture the class of “easy” sentences by a combinatorics. Ideas of relativization have also appeared in \cite{RelGeneral,AtseriasMullerOliva13}. The proof, in fact refutation, system we consider is Resolution-with-bounded-conjunction, denoted Res$(d)$ and introduced by Kraj\'{\i}\v{c}ek in \cite{Krajicek's_Book}. It is an extension of Resolution in which conjunctions of up to $d$ literals are allowed instead of single literals. The tree-like version of Res$(d)$ is usually denoted Res$^* (d)$.  Kraj\'{\i}\v{c}ek proved that tree-like Resolution, and even Res$^* (d)$, have combinatorics associated with it. This follows also from Riis's complexity gap theorem for tree-like
Resolution \cite{complexity_gap}, and shows that the sentences, easy for tree-like Resolution, remain
easy after having been relativized.

The next natural system to look at is Resolution. It is stronger than Res$^* (d)$ for
any $d$, $1 \leq d \leq n$ (equivalent to Res$^* (n)$, in fact, where $n$ is the number of variables), and yet weak enough so that one could expect that it can easily prove some property of the whole universe, but cannot prove it for an arbitrary subset. 
As we show in the paper, this is indeed the case. The example is very natural, the \emph{Least number principle}, LNP$_n$. The contradiction LNP$_n$ asserts that a partial $n$-order has no minimal element. In the literature it enjoys a myriad of alternative names: the  \emph{Graph ordering principle} GOP, \emph{Ordering principle} OP and \emph{Minimal element principle} MEP. Where the order is total it is also known as TLNP and GT.
It is not hard to see that LNP$_n$ is easy for Resolution \cite{Maria}, and we prove that its relativization RLNP$_n$ is hard. A more general result has been proven in \cite{Switching_small_restrictions_journal}; however the lower bound there is weaker. We also consider iterated relativization, and show that the $d$th iteration $d$-RLNP$_n$ is hard for Res$(d)$, but easy for Res$(d+1)$.
We go on to consider the relativization question for Res$^* (d)$, where the FO language
is enriched with a built-in order. The complexity gap theorem does not hold in this
setting \cite{RelGeneral}, though we are able to show that relativization again makes some sentences harder. A variant of the \emph{Induction Principle} gives the contradiction IP$_n$, saying that there is a property which: holds for the minimal element; if it holds for a particular element, there is a bigger one for which the property holds, too; and the property does not hold for the maximal element. We prove that the $d$th iteration of the relativization of the Induction principle, $d$-RIP$_n$, is easy for Res$^* (d+1)$, but hard for Res$^* (d)$.
More precisely, our results are the following:
\begin{itemize}
\item[1.] Any Resolution refutation of RLNP$_n$
is of size $2^{\Omega(n)}$. Firstly, this answers positively to Kraj\'{\i}\v{c}ek's question. Secondly,
observing that RLNP$_n$ has an $O(n^3)$-size refutation in Res$(2)$, we get an exponential separation
between Resolution and Res$(2)$. A similar result was proved in \cite{Switching_small_restrictions_journal} (see
also \cite{AutomatizabilityResolution} for a weaker, quasi-polynomial, separation). Our proof is quite simple compared with that of \cite{Switching_small_restrictions_journal}, where this separation is a corollary of a more general result, and our lower bound is stronger.
\item[2.] $d$-RLNP$_n$ has an $O(dn^3)$-size refutation in Res$(d+1)$, but requires $2^{\Omega(n^\epsilon)}$-size refutation in Res$(d)$, where $\epsilon$ is a constant dependent on $d$. These separations were first proved in \cite{Switching_small_restrictions_journal}. As a matter of fact, we use their method but our tautologies are more natural, and our proof is a little simpler.
\item[3.] $d$-RIP$_n$ has an $O(dn^2)$-size Res$^* (d+1)$ refutation, but requires Res$^* (d)$ refutations of size $2^{\Omega(\frac{n}{d}})$. This holds for any $d$, $0 \leq d \leq n$. A similar result was proven in \cite{Tree-like_Res(k)}. Again, our tautologies are more natural, while the proof is simpler.
\end{itemize}
The second part of the paper is in the area of \emph{Parameterized proof complexity}, a program initiated in \cite{FOCS2007}, which generally aims to gain evidence that $\W[i]$ is different from $\FPT$. Typically, $i$ is so that the former is $\W[2]$, though---in the journal version \cite{FOCS2007journal} of \cite{FOCS2007}---this has been $\WSAT$ and---in the note \cite{PPCw1}---$\W[1]$ was entertained.
In the $\W[2]$ context, parameterized refutation systems aim at refuting \emph{parameterized contradictions} which are pairs $(\mathcal{F},k)$ in which $\mathcal{F}$ is a propositional CNF with no satisfying assignment of weight $\leq k$. Several parameterized (hereafter often abbreviated as ``p-'') proof systems are discussed in \cite{FOCS2007,Galesietal,BGLRjournal}. The lower bounds in  \cite{FOCS2007}, \cite{Galesietal} and \cite{BGLRjournal} amount to proving that the systems p-tree-Resolution, p-Resolution and p-bounded-depth Frege, respectively, are not \emph{fpt-bounded}. Indeed, this is witnessed by the \emph{Pigeonhole principle}, and so holds even when one considers parameterized contradictions $(\mathcal{F},k)$ where $\mathcal{F}$ is itself an actual contradiction. Such parameterized contradictions are termed ``\emph{strong}'' in \cite{BGLRjournal}, in which the authors suggest these might be the only parameterized contradictions worth considering, as general lower bounds---even in p-bounded-depth Frege---are trivial (see \cite{BGLRjournal}). We sympathize with this outlook, but remind that there are alternative parameterized refutation systems built from embedding (see \cite{FOCS2007,FOCS2007journal}) for which no good lower bounds are known even for general parameterized contradictions.

In the world of parameterized proof complexity, we already have lower bounds for $\mathrm{p\mbox{-}Res}(d)$ (as we have for p-bounded-depth Frege), but we are still interested in separating levels $\mathrm{p\mbox{-}Res}(d)$. We are again able to use the relativized Least number principle, RLNP$_n$ to separate p-Res$(1)$ and p-Res$(2)$. Specifically, we prove that \begin{itemize}
\item[4.] $(\mathrm{RLNP}_n,k)$ admits a polynomial-sized in $n$ refutation in Res$(2)$, but all p-Res$(1)$ refutations of $(\mathrm{RLNP}_n,k)$ are of size $\geq n^{\sqrt{(k-3)/16}}$. 
\end{itemize}
\noindent Although we use the same principle as in the first part of this paper, the classical proof of bottleneck counting does not adapt to the parameterized world, and instead we look for inspiration to the proof given in \cite{BGLRjournal} for the \emph{Pigeonhole principle}. For tree-like Resolution, the situation is more complicated. RIP$_n$, admits fpt-bounded proofs in Res$^*(1)$, indeed of size $O(k!)$, therefore we are forced to alter this principle. Thus we come up with the \emph{relativized vectorized induction principle} RVIP$_n$. We are able to show that 
\begin{itemize}
\item[5.] $(\mathrm{RVIP}_n,k)$ admits $O(n^4)$ refutations in Res$^*(2)$, while every refutation in Res$^*(1)$ is of size $\geq n^{k/16}$. 
\end{itemize}
\noindent Note that both of our parameterized contradictions are ``strong'', in the sense of \cite{BGLRjournal}. We go on to give extended versions of RVIP$_n$ and explain how they separate p-Res$^*(d)$ from p-Res$^*(d+1)$, for $d>1$.

This paper is organized as follows. We begin with the preliminaries.
% in which we also introduce a relationship between Res$^*(d)$ and Resolution with bounded treewidth $d$. 
After this, we give the results about Resolution and Res$(d)$ complexity of the $d$-RNLP$_n$ (Items 1 and 2 above) in Section~\ref{MEP-Res}. We then give the results about Res$^*(d)$ complexity of the $d$-RIP$_n$ (Item 3) in Section~\ref{IP-TlRes}. Moving to the parameterized world, we give our separation of p-Res$(1)$ from p-Res$(2)$ (Item 4) in Section~\ref{sec:res}, and our separations of p-Res$^*(d)$ from p-Res$^*(d+1)$ (Item 5) in Section~\ref{sec:tree-res}. We then conclude with some final remarks and open questions.

\section{Preliminaries}

We use the notation $[n]:=\{1,\ldots,n\}$ and we denote by $\top$ and $\bot$ the Boolean values \emph{true} and \emph{false}, respectively. A \emph{literal} is either a propositional variable or its negation. A $d$-conjunction ($d$-disjunction) is a conjunction (disjunction) of at most $d$ literals. A \emph{term} ($d$-\emph{term}) is either a conjunction ($d$-conjunction) or a constant, $\top$ or $\bot$. A $d$-\emph{DNF} is a disjunction of (an unbounded number) of $d$-conjunctions. 

A $d$-\emph{CNF} is a conjunction of (an unbounded number) of $d$-disjunctions. Thus we may identify CNFs with sets of clauses. Since variables will often be written in Roman capitals, conjuncts, disjuncts, CNFs and DNFs will benefit from being written calligraphically.

As we are interested in translating FO sentences into sets of clauses, we assume that
a finite $n$-element universe $\mathbb{U}$ is given. The elements of $\mathbb{U}$ are the first $n$ positive natural
numbers, i.e. $\mathbb{U}:=[n]$. When we say ``element'' we always assume an element from
the universe. We will not explain the translation itself; the details can be found in \cite{uniform_tautologies} or \cite{Krajicek's_Book}. An example of this translation can be found at the beginning of Section~\ref{MEP-Res}.

\subsection{Resolution and Res$(d)$}

The system of Resolution aims to refute a set of clauses by inferring from them the empty clause (a logical contradiction). We will introduce this system through its generalization, due to  Kraj\'{\i}\v{c}ek \cite{Krajicek's_Book}, to Res$(d)$.

Res$(d)$ is a system to refute a set of $d$-DNFs. There are four derivation rules. The $\wedge$-\emph{introduction rule} allows one to derive from $\mathcal{P} \vee \bigwedge_{i\in I_{1}}\ell_{i}$ and $\mathcal{Q} \vee \bigwedge_{i\in I_{2}}\ell_{i}$, $\mathcal{P} \vee \mathcal{Q} \vee \bigwedge_{i\in I_{1} \cup I_2} \ell_{i}$, provided $|I_1 \cup I_2|\leq d$ ($\mathcal{P}$ and $\mathcal{Q}$ are $d$-DNFs). The \emph{cut} (or \emph{resolution}) rule allows one to derive from $\mathcal{P} \vee \bigvee_{i\in I}\ell_{i}$ and $\mathcal{Q} \vee \bigwedge_{i\in I}\neg \ell_{i}$, $\mathcal{P} \vee \mathcal{Q}$. Finally, the two weakening rules allow the derivation of  $\mathcal{P} \vee \bigwedge_{i\in I}\ell_{i}$ from $\mathcal{P}$, provided $|I| \leq d$, and $\mathcal{P} \vee \bigwedge_{i\in I_{1}}\ell_{i}$ from $\mathcal{P} \vee \bigwedge_{i\in I_{1} \cup I_2}\ell_{i}$. 

A Res$(d)$ refutation can be considered as a directed acyclic graph (DAG), whose sources
are the initial clauses, called also axioms, and whose only sink is the empty clause. We
will measure the size of a refutation as the number of internal nodes of the graph, i.e.
the number of applications of a derivation rule.
Whenever we say ``we refute an FO sentence in Res$(d)$'', we mean that we first
translate the sentence into a set of clauses defined on a finite universe of size $n$, and
then refute it with Res$(d)$. The size of the refutation is then a function in $n$.
In the notation ``Res$(d)$'', $d$ may be seen as a function in the number of propositional
variables. Important special cases are Res$(\log)$ as well as Res(const).
Clearly Res$(1)$ is (ordinary) Resolution. In this
case, we have only usual clauses, \mbox{i.e.} disjunctions of literals. The cut rule becomes the usual resolution rule, and only the first weakening rule is meaningful. The letter $d$ is reserved for use as relating to the parameter of some Res$(d)$. 

\subsection{Res$(d)$ as a branching program}

If we turn a Res$(d)$ refutation of a given set of $d$-DNFs $\Sigma$ upside-down, i.e. reverse the edges of the underlying graph and negate the $d$-DNFs on the vertices, we get a special kind of \emph{restricted branching $d$-program}. We introduce them in generality (not just for $d=1$) as we will use them in the part of the paper on parameterized proof complexity. The restrictions are
as follows.
Each vertex is labelled by a $d$-CNF which partially represents the  
information
that can be obtained along any path from the source to the vertex (this is a \emph{record} in the parlance of \cite{proofs_as_games}).
Obviously, the (only) source is labelled with the constant $\top$.
There are two kinds of queries, which can be made by a vertex:
\begin{enumerate}
\item Querying a new $d$-disjunction, and branching on the answer: that is, from $\mathcal{C}$ and the question $\bigvee_{i\in I}\ell_{i}?$ we split on $\mathcal{C}\wedge\bigvee_{i\in I}\ell_{i}$ and  $\mathcal{C}\wedge\bigwedge_{i\in I} \neg \ell_{i}$.
\item Querying a known $d$-disjunction, and splitting it according to  
the answer: that is, from $\mathcal{C} \wedge \bigvee_{i\in I_1 \cup I_2}\ell_{i}$ and the question $\bigvee_{i\in I_1}\ell_{i}?$ we split on $\mathcal{C} \wedge \bigvee_{i\in I_1}\ell_{i}$ and  $\mathcal{C} \wedge \bigvee_{i\in I_2}\ell_{i}$.
\end{enumerate}
\noindent There are two ways of forgetting information. From $\mathcal{C}_1 \wedge \mathcal{C}_2$ we can move to $\mathcal{C}_1$. And from $\mathcal{C} \wedge \bigvee_{i\in I_1}\ell_{i}$ we can move to $\mathcal{C} \wedge \bigvee_{i\in I_1 \cup I_2}\ell_{i}$. The point is that forgetting allows us to equate the information obtained along two different branches and thus to merge them into a single new vertex. A sink of the branching $d$-program must be labelled with the negation of a $d$-DNFs from $\Sigma$. Thus the branching $d$-program is supposed by default to solve the \emph{Search problem for $\Sigma$}: given an assignment of the variables, find a $d$-DNF which is falsified under this assignment.

The equivalence between a Res$(d)$ refutation of $\Sigma$ and a branching $d$-program of the kind above is obvious. If we do not allow the forgetting of information, we will not be able to merge distinct branches, so what we get is a class of decision trees that correspond precisely to the tree-like version of these refutation systems. Indeed, a tree-like branching $2$-program is depicted later in the paper in Figure~\ref{fig:figure1}. 
Naturally, if we allow querying single variables only, we get branching $1$-programs---decision DAGs---that correspond to Resolution.
These decision DAGs permit the view of Resolution as a game between a Prover and Adversary (originally due to Pudlak in \cite{proofs_as_games}). Playing from the unique source, Prover questions variables and Adversary answers either that the variable is true or false (different plays of Adversary produce the DAG). Internal nodes are labelled by conjunctions of facts (\emph{records} to Pudlak) and the sinks hold conjunctions that contradict an initial clause. Prover may also choose to forget information at any point---this is the reason we have a DAG and not a tree. Of course, Prover is destined to win any play of the game---but a good Adversary strategy can force that the size of the decision DAG is large, and many Resolution lower bounds have been expounded this way. 

In order to prove our lower bounds on Resolution refutations, we will use the well-known and classical technique of \emph{bottleneck counting}. This was introduced by Haken in his seminal paper \cite{Haken's_classical} (for the modern treatment see \cite{proofs_as_games}). We first define the concept of \emph{big clause}. We then design random restrictions, so that they ``kill'' (\mbox{i.e.} evaluate to $\top$) any big clause with high probability (w.h.p.). By the union bound, if there are few big clause, there is a restriction which kills them all. We now consider the restricted set of clauses, and using the Prover-Adversary game, show that there has to be at least one big clause in the
restricted proof, which is a contradiction that completes the argument.

The case of Res$(d)$ is not so easy. A general method for proving lower bounds is
developed in \cite{Switching_small_restrictions_journal}. We first hit the refutation by random restrictions, such that all the $d$-DNFs in the refutation, under the restrictions, can be represented by shallow Boolean decision trees w.h.p.  We then use the fact, proved in \cite{Switching_small_restrictions_journal}, that such a refutation can be transformed
into a small width Resolution refutation. Finally we consider the restricted set of clauses, and using the Prover-Adversary game, show that there has to be at least one big clause in the Resolution refutation. This gives the desired contradiction to the assumption that the initial Res$(d)$ refutation contains a small number of $d$-DNFs.

The case of tree-like refutations, either Resolution or Res$(d)$, is much simpler, as a tree-like refutation
of a given set of clause is equivalent to a decision tree, solving the search
problem. We can then use a quite straightforward adversary argument against a decision tree, in order to show that it has to have many nodes. Adversary will play to a strategy that occasionally permits him to give Prover a free choice, this allows the branching in a subtree that gives a lower bound on that for the refutation.

\subsection{Parameterized refutation systems}

A \emph{parameterized language} is a language $L\subseteq A^* \times \mathbb{N}$ where $A$ is an alphabet; in an instance $(x,k) \in L$, we refer to $k$ as the \emph{parameter}. A parameterized language is \emph{fixed-parameter tractable} (fpt and in \FPT) if membership in $L$ can be decided in time $f(k)|x|^{O(1)}$ for some computable function $f$. If FPT is the parameterized analog of P, then (at least) an infinite chain of classes vie for the honour to be the analog of NP. The so-called W-hierarchy sits thus: $\FPT \subseteq \W[1] \subseteq \W[2] \subseteq \ldots \subseteq \WSAT$. For more on parameterized complexity and its theory of completeness, we refer the reader to the monographs \cite{DowneyFellows,FlumGrohe}. Recall that the \emph{weight} of an assignment to a propositional formula is the number of variables evaluated to true. Of particular importance to us is the parameterized problem \textsc{Bounded-CNF-Sat} whose input is $(\mathcal{F},k)$ where $\mathcal{F}$ is a formula in CNF and whose yes-instances are those for which there is a satisfying assignment of weight $\leq k$. \textsc{Bounded-CNF-Sat} is complete for the class $\W[2]$, and its complement (modulo instances that are well-formed formulae) \textsc{PCon} is complete for the class co-$\W[2]$. Thus, \textsc{PCon} is the language of \emph{parameterized contradictions}, $(\mathcal{F},k)$ \mbox{s.t.} $\mathcal{F}$ is a CNF which has no satisfying assignment of weight $\leq k$.  The letter $k$ is reserved for use as pertaining to this weight bound. 

A \emph{proof system} for a parameterized language $L \subseteq A^* \times \mathbb{N}$ is a poly-time computable function $P:A^* \rightarrow A^*\times \mathbb{N}$ \mbox{s.t.} $\mathrm{range}(P)=L$. $P$ is \emph{fpt-bounded} if there exists a computable function $f$ so that each $(x,k)\in L$ has a proof of size at most $f(k)|x|^{O(1)}$. 
These definitions come from \cite{Galesietal,BGL-SAT,BGLRjournal} and are slightly different from those in \cite{FOCS2007,FOCS2007journal} (they are less unwieldy and have essentially the same properties). The program of \emph{parameterized proof complexity} is an analog of that of Cook-Reckow \cite{Proof_Complexity_start}, in which one seeks to prove results of the form $\W[2]\neq$co-$\W[2]$ by proving that parameterized proof systems are not fpt-bounded. This comes from the observation that there is an fpt-bounded parameterized proof system for a co-$\W[2]$-complete $L$ if $\W[2]=$co-$\W[2]$.

The system of \emph{parameterized Resolution} \cite{FOCS2007} seeks to refute the parameterized contradictions of \textsc{PCon}. Given $(\mathcal{F},k)$, where $\mathcal{F}$ is a CNF in variables $x_1,\ldots,x_n$, it does this by providing a Resolution refutation of 
\begin{equation}
\mathcal{F}\cup \{\neg x_{i_1}\vee \ldots \vee \neg x_{i_{k+1}} : 1 \leq i_1 < \ldots < i_{k+1} \leq n \}.
\label{equ:W[2]}
\end{equation}
Thus, in parameterized Resolution we have built-in access to these additional clauses of the form $\neg x_{i_1}\vee \ldots \vee \neg x_{i_{k+1}}$, but we only count those that appear in the refutation. We may consider any refutation system as a parameterized refutation system, by the addition of the clauses given in (\ref{equ:W[2]}). In particular, parameterized Res$(d)$, p-Res$(d)$, will play a part in the sequel.

\section{Relativized Least number principle and Res$(d)$}
\label{MEP-Res}

%\subsection{Minimal Element Principle ($MEP_{n}$) and relativization}

%%%%
\noindent The \emph{Least number principle}, 
states that a (partial) order, defined on a finite set of $n$ elements,
has a minimal element. Its negation LNP can be expressed as the following
FO sentence:\begin{eqnarray}
 &  & \left(\left(\forall x\: \neg L\left(x,x\right)\right)\wedge \right.\nonumber \\
 &  & \left(\forall x,y,z\: \left(L\left(x,y\right)\wedge L\left(y,z\right)\right)\rightarrow L\left(x,z\right)\right)\wedge \label{eq:MEP}\\
 &  & \left.\left(\forall x\exists y\: L\left(y,x\right)\right)\right).\nonumber 
\end{eqnarray}
Here $L\left(x,y\right)$ stands for $x<y$. The encoding of LNP$_{n}$ as a set of clauses
is as follows.
\begin{eqnarray*}
\neg L_{i,i} & i \in [n] \\
\neg L_{i,j} \vee \neg L_{j,\ell} \vee L_{i,\ell} & i,j,\ell \in [n] \\  
\bigvee_{i \in [n]} S_{i,j} & j \in [n] \\
\neg S_{i,j} \vee L_{i,j} & i,j \in [n] \\
\end{eqnarray*}
where $S$ is the Skolem relation, witnessing the existential variable $y$ from $\forall x\exists y$  $L\left(y,x\right)$,
i.e., for each $j$, $S_{i,j}=\top $ implies that the $i$th element
is smaller than the $j$th one. Of course, this $S$ relation is unnecessary in the standard LNP$_{n}$ (one may remove it and replace each $\bigvee_{i \in [n]} S_{i,j}$ with $\bigvee_{i \in [n]} L_{i,j}$). However, it will become necessary in the relativizations which we now introduce.

\

\noindent The negation of the \emph{$d$-Relativized Least number principle}, $d$-RLNP$_{n}$,
is as follows. Let $R^{p}$, $1\leq p\leq d$, be the unary predicates
which we relativize by, and let us denote by $\mathcal{R}\left(x\right)$
the conjunction $\bigwedge _{p\in \left[d\right]}R^{p}\left(x\right)$.
 $d$-RLNP$_{n}$ is the following sentence:\begin{eqnarray*}
 &  & \left(\left(\forall x\: \mathcal{R}\left(x\right)\rightarrow \neg L\left(x,x\right)\right)\wedge \right.\\
 &  & \left(\forall x,y,z\: \mathcal{R}\left(x\right)\wedge \mathcal{R}\left(y\right)\wedge \mathcal{R}\left(z\right)\rightarrow \right.\\
 &  & \left.\left(L\left(x,y\right)\wedge L\left(y,z\right)\right)\rightarrow L\left(x,z\right)\right)\wedge \\
 &  & \left(\forall x\exists y\: \mathcal{R}\left(x\right)\rightarrow \left(\mathcal{R}\left(y\right)\wedge L\left(y,x\right)\right)\right)\wedge \\
 &  & \left.\left(\exists x\: \mathcal{R}\left(x\right)\right)\right).
\end{eqnarray*}
What this is saying is that the negation of the least number principle holds on the subuniverse given by  $\mathcal{R}(x)$, and this subuniverse is non-empty. The corresponding translation into clauses (simplified by assuming the witness to the final $\mathcal{R}$ be $n$) gives the following.
\begin{eqnarray*}
\neg \mathcal{R}_i \vee \neg L_{i,i} & i \in [n] \\
\neg \mathcal{R}_i \vee \neg \mathcal{R}_j \vee \neg \mathcal{R}_\ell \vee \neg L_{i,j} \vee \neg L_{j,\ell} \vee L_{i,\ell} & i,j,\ell \in [n] \\  
\bigvee_{i \in [n]} S_{i,j} & j \in [n] \\
\neg S_{i,j} \vee \neg \mathcal{R}_j \vee R^p_i & i, j \in [n]; p \in [d] \\
\neg S_{i,j} \vee \neg \mathcal{R}_j \vee L_{i,j} & i, j \in [n] \\
R^p_n & p \in [d] \\
% \label{RExistentialClauses}
\end{eqnarray*}
\noindent We generally write RLNP$_{n}$ for  $1$-RLNP$_{n}$.

\subsection{The upper bound: $d$-RLNP$_{n}$ is easy for Res$\left(d+1\right)$}

\begin{proposition} There is an $O\left(dn^{3}\right)$ size Res$\left(d+1\right)$
refutation of $d$-RLNP$_{n}$.
\end{proposition}
\begin{proof}
As in the previous subsection, $\mathcal{R}_{i}$ is the $d$-conjunction $\bigwedge _{p\in \left[d\right]}R_{i}^{p}$; clearly $\neg \mathcal{R}_{i}$ is then a $d$-disjunction.

The Res$\left(d+1\right)$ proof will consists of $n$ stages. We
will show how to construct it, starting from the $n$th stage, and
going to the $1$st one.

The $\ell$th stage clauses that we will need to derive are
\begin{eqnarray}
\neg \mathcal{R}_{j}\vee \bigvee _{i \in \left[\ell\right],\: j \neq i}\left(L_{i,j}\wedge \mathcal{R}_{i}\right) &  & j \in \left[\ell\right]\label{Stage-k-minClauses}
\end{eqnarray}
together with\begin{equation}
\bigvee _{i \in \left[\ell\right]}\mathcal{R}_{i}.\label{Stage-k-EmptyClause}\end{equation}
 Thus the $1$st stage clauses are $\neg \mathcal{R}_{1}$ and $\mathcal{R}_{1}$
which, resolved, give the empty clause. An $n$th stage clause of
the form (\ref{Stage-k-minClauses}) can be derived from the axioms. We first use $d \left(n-1\right)$ applications of $\wedge $-introductions
to derive $\neg S_{i,j} \vee \neg \mathcal{R}_{j}\vee \vee \left(\mathcal{R}_{i}\wedge L_{i,j}\right)$,
$j \in \left[n\right]$, $j \neq i$, and then $n-1$ resolutions of
variables $S_{i,j}$. Finally we {}``kill'' the literal $S_{i,i}$
by two resolutions with the axioms $\neg S_{j,j}\vee \neg \mathcal{R}_{j}\vee L_{j,j}$
and $\neg \mathcal{R}_{j}\vee \neg L_{j,j}$.

We derive the $n$th stage clause of the form (\ref{Stage-k-EmptyClause}) by directly weakening the axioms $\mathcal{R}_n$.

What remains is to show how to derive the $\left(\ell-1\right)$th stage
clauses from the $\ell$th stage ones. The clause 
\begin{eqnarray}
\label{Stage-k-1-minClauses}
\neg \mathcal{R}_{j}\vee \bigvee _{i \in \left[\ell-1\right],\: j \neq i}\left(L_{i,j}\wedge \mathcal{R}_{i}\right),\quad j \in \left[\ell-1\right]
\end{eqnarray}
can be derived as follows. We start off with the $\ell$th stage clause\[
\neg \mathcal{R}_{\ell}\vee \bigvee _{i \in \left[\ell-1\right]}\left(L_{i,\ell}\wedge \mathcal{R}_{i}\right),
\]
where we can assume the big disjunction omits $i=\ell$ by Resolution with the irreflexivity axiom  $\neg \mathcal{R}_{\ell}\vee \neg L_{\ell,\ell}$, and manipulate it at $\ell-1$ substages, the $i$th substage dealing
with the conjunction $L_{i,\ell}\wedge \mathcal{R}_{i}$, $i \in \left[\ell-1\right]$.
Let us consider the clause before the $i$th substage, and denote
it by\begin{equation}
\mathcal{C}\vee \left(L_{i,\ell}\wedge \mathcal{R}_{i}\right)\label{SubstageInitialClause}\end{equation}
(here $\mathcal{C}$ is the corresponding subclause). We first resolve
it with the transitivity axiom
%\[ \neg \mathcal{R}_{i}\vee \neg \mathcal{R}_{j}\vee \neg \mathcal{R}_{\ell}\vee \neg L_{j,\ell}\vee \neg L_{\ell,i}\vee L_{j,i}\]
\[ \neg \mathcal{R}_{i}\vee \neg \mathcal{R}_{\ell}\vee \neg \mathcal{R}_{j}\vee \neg L_{i,\ell}\vee \neg L_{\ell,j}\vee L_{i,j}\]
to get
%\[\neg \mathcal{R}_{i}\vee \neg \mathcal{R}_{\ell}\vee \neg L_{\ell,i}\vee \mathcal{C}\vee L_{j,i},\]
\[
\neg \mathcal{R}_{j}\vee \neg \mathcal{R}_{\ell}\vee \neg L_{\ell,j}\vee \mathcal{C}\vee L_{i,j},\]
and then apply a $\wedge $-introduction with $\mathcal{C}\vee \mathcal{R}_{i}$,
which is a weakening of (\ref{SubstageInitialClause}) to get
%\[\neg \mathcal{R}_{i}\vee \neg \mathcal{R}_{\ell}\vee \neg L_{\ell,i}\vee \mathcal{C}\vee \left(L_{j,i}\wedge \mathcal{R}_{j}\right).\]
\[
\neg \mathcal{R}_{j}\vee \neg \mathcal{R}_{\ell}\vee \neg L_{\ell,j}\vee \mathcal{C}\vee (L_{i,j} \wedge \mathcal{R}_i).\]
Thus after having completed the $\ell-1$ substages we get the clause
%\[\neg \mathcal{R}_{i}\vee \neg \mathcal{R}_{\ell}\vee \neg L_{\ell,i}\vee \bigvee _{j \in \left[\ell-1\right]}\left(L_{j,i}\wedge \mathcal{R}_{j}\right).\]
\[
\neg \mathcal{R}_{j}\vee \neg \mathcal{R}_{\ell}\vee \neg L_{\ell,j}\vee \bigvee_{i \in [\ell-1]}(L_{i,j} \wedge \mathcal{R}_i).\]
A resolution step with the irreflexivity axiom $\neg \mathcal{R}_{j}\vee \neg L_{j,j}$
``kills'' the conjunction $L_{j,j}\wedge \mathcal{R}_{j}$, and we
get
% \[\neg \mathcal{R}_{i}\vee \neg \mathcal{R}_{\ell}\vee \neg L_{\ell,i}\vee \bigvee _{j \in \left[\ell-1\right],\: j \neq i}\left(L_{j,i}\wedge \mathcal{R}_{j}\right).\]
\[
\neg \mathcal{R}_{j}\vee \neg \mathcal{R}_{\ell}\vee \neg L_{\ell,j}\vee \bigvee_{i \in [\ell-1], i\neq j}(L_{i,j} \wedge \mathcal{R}_i).\]
A final cut with another $\ell$th stage clause,\[
\neg \mathcal{R}_{j}\vee \bigvee _{i \in \left[\ell \right],\: j \neq i}\left(L_{i,j}\wedge \mathcal{R}_{i}\right),\]
now gives the desired result (\ref{Stage-k-1-minClauses}).

The $\left(\ell-1\right)$th stage clause $\bigvee _{i \in \left[\ell-1\right]}\mathcal{R}_{i}$
is easier to derive. We weaken the $\ell$th stage clause $\neg \mathcal{R}_{\ell}\vee \bigvee _{i \in \left[\ell-1\right]}\left(L_{i,\ell}\wedge \mathcal{R}_{i}\right)$
$\ell-1$ times to get $\neg \mathcal{R}_{\ell}\vee \bigvee _{i \in \left[\ell-1\right]}\mathcal{R}_{i}$.
We then resolve it with the $\ell$th stage clause $\bigvee _{i \in \left[\ell \right]}\mathcal{R}_{i}$.

This completes the proof.
\end{proof}

%%%%
\subsection{An optimal lower bound: RLNP$_n$ is exponentially hard for Resolution}

We will prove the following using the well-known method of random restrictions.

\begin{proposition}
\label{ResLBnd}
Any Resolution proof of RLNP$_n$ is of size $2^{\Omega \left(n\right)}$.
\end{proposition}
\begin{proof}
 The idea is to randomly divide the universe
$\mathbb{U}$ into two approximately equal parts. One of them, $\mathbb{R}$,
will represent the predicate $R$; all the variables within it will
remain unset. The rest, $\mathbb{C}$, will be the {}``chaotic''
part; all the variables within $\mathbb{C}$ and most of the variables
between $\mathbb{C}$ and $\mathbb{R}$ will be set at random. It
is now intuitively clear that while $\mathbb{C}$ kills with positive
probability a certain number of {}``big'' clauses, $\mathbb{R}$
allows to show, via an adversary argument, that at least one such
clause must be present in any Resolution refutation, after it has
been hit by the random restrictions. Therefore a huge number of {}``big''
clauses must have been presented in the original refutation.

We of course keep $R_n:=\top$. The \emph{random restrictions} are as follows. 

\begin{enumerate}
\item We first set all the variables $R_{i}$, $i\in \left[n-1\right]$,
to either $\top $ or $\bot $ independently at random with equal
probabilities, $1/2$. Let us denote the set of variables with $R_{i}=\top $
by $\mathbb{R}$, and the set of variables with $R_{i}=\bot $ by
$\mathbb{C}$, $\mathbb{C}=\mathbb{U}\setminus \mathbb{R}$. 
\item We now set all the variables $L_{i,j}$ with at least one endpoint
in $\mathbb{C}$, i.e. $\left\{ i,j\right\} \cap \mathbb{C}\neq \emptyset$,
to either $\top $ or $\bot $ independently at random with equal
probabilities, $1/2$.
\item For each $j\in \mathbb{C}$, $j\neq i$ we set $S_{i,j}$ to either
$\top $ or $\bot $ independently at random with equal probabilities
$1/2$. Note that it is possible to set all the $S_{i,j}$ to $\bot $,
thus violating an axiom. It however happens with small probability,
$1/2^{n-1}$ for a fixed $j$.
\item We finally set all the variables $S_{i,j}$ with $i\in \mathbb{C}$,
$j\in \mathbb{R}$ to $\bot $.
\end{enumerate}
Note the unset variables define exactly the non-relativized principle
on $\mathbb{R}$, LNP$_{\left|\mathbb{R}\right|}$.

By the Chernoff bound (see \cite{RandomizedAlgorithms}) the probability that $\mathbb{R}$ contains
less than $n/4$ elements is exponentially small, and we therefore have the following.
\begin{observation}
\label{GoodRandomRestrictions}
The probability that the
random restrictions are inconsistent (\mbox{i.e.} violate an axiom) or $\left|\mathbb{R}\right|\leq n/4$
is at most $\left(\textrm{n-1}\right)2^{-\left(n-1\right)}+e^{-n/16}$.
\end{observation}

A \emph{big clause} is one which contains at least $n/8$ busy elements.
The element $i$ is \emph{busy} in the clause $\mathcal{C}$ iff $\mathcal{C}$
contains one of the following variables, either positively or negatively:
$R_{i}$, $L_{i,j}$, $L_{j,i}$, $S_{j,i}$ for some $j$, $j\neq i$ (note the omission of $S_{i,j}$).
We can now deduce the following.
\begin{observation}
\label{BigClauseCollapse}
A clause, containing $p$ busy
elements, does not evaluate to $\top $ under the random restrictions
with probability at most $\left(3/4\right)^{p/2}$. 
\end{observation}

Indeed, let us consider the different cases of a busy element $i$
in the clause $\mathcal{C}$:

\begin{enumerate}
\item The variable $R_{i}$ is present in $\mathcal{C}$: The probability
that the corresponding term does not evaluate to $\top $ is $1/2$. \emph{}
\item The variable $S_{j,i}$ for some $i \neq j$ is present in $\mathcal{C}$:
The corresponding term does not evaluate to $\top $ if either $i \in \mathbb{R}$, or $i \in \mathbb{C}$ and it evaluates to
$\bot $. The probability of this is $3/4$.
\item Either the variable $L_{i,j}$ or the variable $L_{j,i}$ for some $j\neq i$
is present in $\mathcal{C}$. Let us denote the set of all such elements $i$ or $j$
by $\mathbb{V}$, $\left|\mathbb{V}\right|=\ell$, and the corresponding
subclause of $\mathcal{C}$ induced by those elements as $\mathcal{E}$. That is, $\mathcal{E}$ contains precisely those atoms of $\mathcal{C}$ that are of the form $L_{i,j}$ or $L_{j,i}$.
 Construct the graph $\mathbb{G}$ with
vertex set $\mathbb{V}$ and edge set $\mathbb{E}$ determined
by the variables $L_{i,j}$, \mbox{i.e.} $\mathbb{E}=\left\{ \left\{ i,j\right\} \mid L_{i,j}\textrm{ is present in }\mathcal{C}\right\} $.
Consider any spanning forest of $\mathbb{G}$. Assume that all the
roots are in $\mathbb{R}$ as this only increases the probability
that $\mathcal{E}$ does not evaluate to $\top $. Going from the
root to the leaves in each tree, we see that the probability that
the corresponding edge does not evaluate to $\top $ is $3/4$ (the
same reason as in the 2nd case). Moreover all the edge variables are
independent from each other, and also there are at most $\ell/2$ roots
(exactly $\ell/2$ iff the forest consists of trees having a root and
a single leaf only). Therefore the probability that the subclause
$\mathcal{E}$ does not evaluate to $\top $ is at most $\left(3/4\right)^{\ell/2}$.
\end{enumerate}
As the events from 1, 2 and 3 are independent for different elements
from $\mathbb{U}$, we have completed the argument for Observation
\ref{BigClauseCollapse}.

We can now present the main argument in the proof. We recall that a big clause is one which contains at least $n/8$ busy elements. Assume there is
a Resolution refutation of RLNP$_{n}$ which contains
less than $\left(4/3\right)^{n/8}$ big clauses. From the Observations
\ref{GoodRandomRestrictions} and \ref{BigClauseCollapse}, using
the union-bound on probabilities, we can conclude that there is a
restriction which is consistent, {}``kills'' all the big clauses
(evaluating them to $\top$), and leaves $\mathbb{R}$ big enough
($\left|\mathbb{R}\right|\geq n/4$).
This is because $\left(\textrm{n-1}\right)2^{-\left(n-1\right)}+e^{-n/16} +\left(4/3\right)^{n/8}\left(3/4\right)^{n/16}<1$.
 Recall that the restricted refutation
is nothing but a Resolution refutation of LNP$_{\left|\mathbb{R}\right|}$
on $\mathbb{R}$. What remains to show is that any such refutation
must contain a big clause which would contradict to the assumption
there were {}``few'' big clauses in the original refutation.

We will consider the Prover-Adversary game for LNP$_{\left|\mathbb{R}\right|}$.
At any time $\mathbb{R}$ is represented as a disjoint union of three
sets, $\mathbb{R}=\mathbb{B}\uplus \mathbb{W}\uplus \mathbb{F}$.
$\mathbb{B}$ is the set of all the elements busy in the current clause.
The elements of $\mathbb{B}$ are always totally ordered. $\mathbb{W}$
is the set of witnesses for some elements in $\mathbb{B}$, \mbox{i.e.} for
each $j\in \mathbb{B}$ there is an element $i\in \mathbb{W}\uplus \mathbb{B}$
such that $S_{i,j}=\top $. We assume that, at any time, any element
of $\mathbb{W}$ is smaller than all the elements of $\mathbb{B}$.
$\mathbb{F}$ is the set of {}``free'' elements. It is obvious how
Adversary maintains these sets in the Prover-Adversary game. When
a variable, which makes an element $i\in \mathbb{W}\uplus \mathbb{F}$
busy, is queried he adds $i$ at the bottom of the totally ordered
set $\mathbb{B}$, answers accordingly, and chooses some $j\in \mathbb{F}$ and moves it to $\mathbb{W}$
setting $S_{j,i}=\top $. When all the variables, which kept an element
$\mathbb{B}$ busy, are forgotten, Adversary removes $i$ from $\mathbb{B}$
and removes the corresponding witness $j$ from $\mathbb{W}$ if it
is there (note that it may be in $\mathbb{B}$, too, in which case
it is not removed). In this way Adversary can maintain the partial
assignment consistent as far as $\mathbb{F}\neq \emptyset$. Note
also that $\left|\mathbb{B}\right|\geq \left|\mathbb{W}\right|$.
Therefore at the moment a contradiction is reached we have $\left|\mathbb{B}\right|\geq \left|\mathbb{R}\right|/2\geq n/8$
as claimed.
\end{proof}

\subsection{General lower bounds: $d$-RLNP$_n$ is subexponentially hard for $Res\left(d\right)$}

We will first give the necessary background from \cite{Switching_small_restrictions_journal}.

\begin{definition}[Definition 3.1, \cite{Switching_small_restrictions_journal}]
A \emph{decision tree} is a rooted binary tree in which every internal
node queries a propositional variable, and the leaves are labelled
by either $\top $ or $\bot $. 

Thus every path from the root to a leaf may be viewed as a partial
assignment. Let us denote by $Br_{v}\left(\mathfrak{T}\right)$, for
$v\in \left\{ \top ,\bot \right\} $, the set of paths (partial assignments)
in the decision tree $\mathfrak{T}$ which lead from the root to a
leaf labelled by $v$.

A decision tree $\mathfrak{T}$ \emph{strongly represents} a DNF $\mathcal{F}$
iff for every $\pi \in Br_{v}\left(\mathfrak{T}\right)$, $\mathcal{F}\upharpoonright _{\pi }=v$.

The \emph{representation height} of $\mathcal{F}$, $h\left(\mathcal{F}\right)$,
is the minimum height of a decision tree strongly representing $\mathcal{F}$.
\end{definition}
\begin{definition}[Definition 3.2, \cite{Switching_small_restrictions_journal}]
Let $\mathcal{F}$ be a DNF, and $\mathbb{S}$ be a set of variables.
We say that $\mathbb{S}$ is a cover of $\mathcal{F}$ iff every conjunction
of $\mathcal{F}$ contains a variable from $\mathbb{S}$. The \textbf{covering
number} of $\mathcal{F}$, $c\left(\mathcal{F}\right)$, is the minimum
size of a cover of $\mathcal{F}$. 
\end{definition}

\begin{lemma}[Corollary 3.4, \cite{Switching_small_restrictions_journal}]
\label{Restr-k-DNF}
Let $d\geq 1$, $\alpha >0$, $1\geq \beta ,\gamma >0$, $s>0$, and
let $\mathscr{D}$ be a distribution on partial assignments such
that for every $d$-DNF $\mathcal{G}$, $\textrm{Pr}_{\rho \in \mathscr{D}}\left[\mathcal{G}\upharpoonright _{\rho }\neq \top \right]\leq \alpha 2^{-\beta \left(c\left(\mathcal{G}\right)\right)^{\gamma }}$.
Then for every $d$-DNF $\mathcal{F}$ :\[
\textrm{Pr}_{\rho \in \mathscr{D}}\left[h\left(\mathcal{F}\upharpoonright _{\rho }\right)\geq s\right]\leq \alpha d2^{-2\left(\beta /4\right)^{d}\left(s/2\right)^{\gamma ^{d}}}.\]
\end{lemma}

\begin{lemma}[Theorem 5.1, \cite{Switching_small_restrictions_journal}]
\label{RestrRes(k)2Res}
Let $\mathcal{G}$ be a set of clauses of width at most $w$. If
$\mathcal{G}$ has a $Res\left(d\right)$ refutation so that for
each line $\mathcal{L}$ of the refutation, of $\Gamma $, $h\left(\mathcal{L}\right)\leq w$,
then $\mathcal{G}$ has a Resolution refutation of width at most
$dw$. 
\end{lemma}

We also need the following construction.
% which is only mentioned in \cite{Switching_small_restrictions}, but whose proof can be found in the full version of the same paper.

\begin{lemma}[Subsections 8.3 and 8.4 in \cite{Switching_small_restrictions_journal}]
\label{GoodGraph}
There is an undirected
graph $\mathfrak{G}\left(\left[n\right],E\right)$ on $n$
vertices and max-degree $\theta \left(\ln n\right)$ such that any
Resolution refutation of LNP$_{n}$, restricted on $\mathfrak{G}$,
is of width $\Omega \left(n\right)$.

LNP$_{n}$, restricted on $\mathfrak{G}$, means that for each element
$i$ the witness $j$ has to be a neighbour of $i$ in $\mathfrak{G}$,
i.e. we set $S_{i,j}=\bot $ whenever $\left\{ i,j\right\} \notin E$.
\end{lemma}

We can now prove the desired result.

\begin{proposition}
For every constant $d\geq 1$ there is a constant $\varepsilon _{d}\in \left(0,1\right]$
such that any Res$\left(d\right)$ refutation of $d$-RLNP$_n$
is of size $2^{\Omega \left(n^{\varepsilon _{d}}\right)}$. 
\end{proposition}
\begin{proof}
 We again denote by $\mathcal{R}_{i}$ the
$d$-conjunction $\wedge _{p\in \left[d\right]}R_{i}^{p}$. We consider
$d$-RNLP$_n$, restricted
on the graph $\mathbb{G}$ from Lemma \ref{GoodGraph}.
% We fix all the other $t_{i}$ ($i\neq n$) to $\bot $. 
Thus we have eliminated all the {}``big''
axioms in the encoding of $d$-RNLP$_n$; the biggest
ones are now of width $\theta \left(\ln n\right)$, and therefore
it would be possible to eventually apply Lemma \ref{RestrRes(k)2Res}.

The \emph{random restrictions} are very similar to the ones in the
proof of Proposition~\ref{ResLBnd}. We fix
$\mathcal{R}_{n}$ (i.e. $R_{n}^{p}=\top $ for all
$p\in \left[d\right]$) to $\top $.
\begin{enumerate}
\item For each $i\in \left[n-1\right]$ and $p\in \left[d\right]$ we set
the variable $R_{i}^{p}$ to either $\top $ or $\bot $, independently
at random with equal probabilities, $1/2$. We denote the set of elements
with $\mathcal{R}_{i}=\top $ by $\mathbb{R}$, and the rest by $\mathbb{C}$,
$\mathbb{C}=\mathbb{U}\setminus \mathbb{R}$. Note that $\mathbb{R}\neq \emptyset$
as always $n\in \mathbb{R}$, and by the Chernoff bound\[
Prob\left[\left|\mathbb{R}\right|\leq \frac{n}{2^{d+1}}\right]\leq e^{-n/2^{d+3}}.\]

\item We now set all the variables $L_{i,j}$ with at least one endpoint
in $\mathbb{C}$, i.e. $\left\{ i,j\right\} \cap \mathbb{C}\neq \emptyset$,
to either $\top $ or $\bot $ independently at random with equal
probabilities, $1/2$.
\item For each $j\in \mathbb{C}$, and $i$, a neighbour of $j$ in $\mathbb{G}$,
we set $S_{i,j}$ to either $\top $ or $\bot $ independently at random
with equal probabilities $1/2$. Note that it is possible to set all
the $S_{i,j}$ to $\bot $, thus violating an axiom. It happens with
probability $1/2^{\deg _{\mathfrak{G}}\left(i\right)}=1/n^{\omega _{0}}$
where $\omega _{0}$ is the constant, hidden in $\theta $-denotation
in Lemma \ref{GoodGraph}, and note that we can choose $\omega _{0}\geq 1$. 
\item We finally set all the variables $S_{i,j}$ with $j\in \mathbb{R}$,
$i\in \mathbb{C}$ to $\bot $.
\end{enumerate}
The unset variables define exactly the non-relativized principle on
$\mathbb{R}$ over $\mathfrak{G}$, LNP$_{\left|\mathbb{R}\right|}$, and we have

\begin{observation}
\label{GoodRandomRestrictions-d}
The probability that
the random restrictions are inconsistent (\mbox{i.e.} violate an axiom) or $\left|\mathbb{R}\right|\leq n/2^{d+1}$
is at most $1/n^{\omega _{0}-1}+ e^{-n/2^{d+3}}$. 
\end{observation}
We will now consider the effect of the random restrictions on a given
$d$-DNF $\mathcal{C}$. We need a few definitions first: We say
that a variable $R_{i}^{p}$ for some $p\in \left[d\right]$ \emph{mentions}
the element $i$; a variable $L_{i,j}$ mentions both $i$ and $j$;
a variable $S_{i,j}$ mentions $j$ only (cf. definition of business in Proposition~\ref{ResLBnd}). A formula mentions the union
of elements mentioned by some variable from the formula. The \emph{element-cover
number} of a $d$-DNF $\mathcal{F}$, $c'\left(\mathcal{F}\right)$,
is the minimum cardinality of a set of elements, such that each element
is mentioned by at least one term of the clause. There is an obvious
connection between $c'\left(\mathcal{F}\right)$ and the cover $c\left(\mathcal{F}\right)$:

\begin{observation}
\label{Var2ElementCover} $c'\left(\mathcal{F}\right)=\Omega \left(\sqrt{c\left(\mathcal{F}\right)}/\log n\right)$.
\end{observation}
Indeed $m$ elements mention $dm$ variables $R_{i}^{p}$, and at
most ${{m \choose 2}}+m\max _{i}\deg _{\mathbb{G}}\left(i\right)$
variables $L_{i,j}$/$S_{i,j}$ which makes $O\left(m^{2}+m\log n\right)$
in total, and therefore $c\left(\mathcal{F}\right)=O\left(\left(c'\left(\mathcal{F}\right)\right)^{2}+c'\left(\mathcal{F}\right)\log n\right)$.

We can now show that any $d$-DNF $\mathcal{F}$ collapses under
random restrictions to a short decision tree w.h.p. Let us first note
that $\mathcal{F}$ mentions at most $\omega _{0}\log n$ variables
$S_{j,n}$. We build up a decision tree by first querying all these.
This contributes $\omega _{0}\log n$ to the height and there are
at most $n^{\omega _{0}}$ leaves which contains $d$-DNFs which
do not mention the $n$th element.

Let us take such a $d$-DNF $\mathcal{G}$. We are now going to use
Lemma \ref{Restr-k-DNF}. Let us perform the following experiment.
Take any term, $\mathcal{T}$, of $\mathcal{G}$. Each literal of
$\mathcal{T}$, containing a variable $R_{i}^{p}$, evaluates to $\top$ with
probability $1/2$. Each literal, containing a variable $S_{j,i}$,
has a probability for $i\in \mathbb{C}$ at least $1/2$; indeed the
rest of $\mathcal{T}$ may contain at most $d-1$ positive appearances
of variables $R_{i}^{p}$ and then the last one, not in $\mathcal{T}$,
{}``decides'' $i\in \mathbb{C}$ with probability $1/2$. Therefore
the literal, containing a variable $S_{j,i}$, evaluates to $\top$ 
with probability at least $1/4$. The same argument applies to a literal,
containing a variable $L_{i,j}$. Thus the probability that $\mathcal{T}$
evaluates to $\top $ under the random restrictions is at least $1/4^{d}$
(the fact that the term contains at most $d$ variables is essential
here; indeed consider the $d+1$-term $\mathcal{R}_{i}\wedge S_{j,i}$:
$\mathcal{R}_{i}=\top $ enforces $S_{j,i}$ either $\bot $ or unset,
and therefore there is no way to evaluate the term to $\top $). On
the other hand $\mathcal{T}$ mentions at most $2d$ elements, so
we can repeat the above procedure (i.e. picking a new term) at least
$c'\left(\mathcal{G}\right)/\left(2d\right)$ times and the probability
that each term does not evaluate to $\top $ is at most $1-1/4^{d}$.
Moreover these $c'\left(\mathcal{G}\right)/\left(2d\right)$ trials
are independent as in each of them only elements, not mentioned so
far, are involved. Therefore the probability that $\mathcal{G}$ does
not evaluate to $\top $ under the random restrictions is at most$\left(1-1/4^{d}\right)^{c'\left(\mathcal{G}\right)/\left(2d\right)}$
which is at most $\left(1-1/4^{d}\right)^{\sqrt{c\left(\mathcal{G}\right)}/\left(2d\log n\right)}$
by Observation \ref{Var2ElementCover}.

We can now apply Lemma \ref{Restr-k-DNF} with $\alpha =1$, $\beta =\omega _{1}/\log n$,
where $\omega _{1}$ is a constant, depending on $d$ only, and $\gamma =1/2$.
We set $s=\omega _{2}n$, where $\omega _{2}$ will be fixed later,
and will depend on $d$ only. What we get by the lemma is that the
probability that $\mathcal{G}$, under the random restrictions, cannot
be represented by a decision tree of height at most $\omega _{2}n$
is at most $\exp \left(-\omega _{3}n^{1/2^{d}}/\left(\log n\right)^{d}\right)$,
where $\omega _{3}$ is a constant, dependent on $d$. Going back to
the initial $d$-DNF $\mathcal{F}$, we see that the probability
$\mathcal{F}$, under the random restrictions, cannot be represented
by a decision tree of height at most $\omega _{2}n+\omega _{0}\log n$
is at most $n^{\omega _{0}}\exp \left(-\omega _{3}n^{1/2^{d}}/\left(\log n\right)^{d}\right)$.

We can finally present the main argument of the proof. Suppose, for
the sake of contradiction, that there is a Res$\left(d\right)$ refutation
of $d$-RLNP$_n$ containing less than $n^{- \omega _{0}}\exp \left(\left(\omega _{3}/2\right)n^{1/2^{d}}/\left(\log n\right)^{d}\right)$
$d$-DNFs. By the union bound, the probability that at least one
of them, under the random restrictions, cannot be represented by a
decision tree of height at most $\omega _{2}n+\omega _{0}\log n$
is at most $\exp \left(-\left(\omega _{3}/2\right)n^{1/2^{d}}/\left(\log n\right)^{d}\right)$.
By Lemma~\ref{RestrRes(k)2Res} then, with the same probability,
the restricted Res$\left(d\right)$ refutation cannot be transformed
to a width $\omega _{2}dn+\omega _{0}d\log n$ Resolution refutation. Adding
up this probability with the probability from Observation~\ref{GoodRandomRestrictions-d},
we see there exists a {}``good'' restriction, \mbox{i.e.} such that the
final Resolution refutation is of width $\omega _{2}dn+\omega _{0}d\log n$.
Recall that we are still free to choose $\omega _{2}$ (it affects
$\omega _{3}$, but it is fine as $\omega _{3}$ is a constant, dependent
on $d$, too) and it depends on $d$ only, so we can ensure that $\omega _{2}$
is smaller than the constant, hidden in the $\Omega $-notation in
Lemma~\ref{GoodGraph}, divided by $2^{d+1}$ (recall that the size
of $\mathbb{R}$ is at least $n/2^{d+1}$). This is the desired contradiction
which completes the proof.
\end{proof}

\section{Relativized Induction principle and Res$^*(d)$}
\label{IP-TlRes}

In this section we consider a version of the Induction principle,
denoted IP$_{n}$, that can be encoded as an FO
sentence if a built-in predicate, defining a total order on the universe,
is added to the language. It is easy to show that IP$_{n}$ is easy
for tree-like Resolution, and so it is also for its $d$th relativization $d$-RIP$_n$, but for Res$^*\left(d+1\right)$. Finally we prove that $d$-RIP$_n$
is hard for Res$^*\left(d\right)$.

\subsubsection{Induction principle}

The (negation of the) Induction principle, we consider, is the following
simple statement: Given an ordered universe, there is a property $P$,
such that 
\begin{enumerate}
\item The property holds for the smallest element.
\item If $P\left(x\right)$ hold for some $x$, then there is $y$, bigger
than $x$, and such that $P\left(y\right)$ holds. 
\item The property does not hold for the biggest element.
\end{enumerate}
The universe $\mathbb{U}$ can now be considered as the set of the first
$n$ natural numbers. In our language we can use the relation symbol
$<$ with its usual meaning. We can also use any
constant $c$ as well as $n-c$ (note that in the language $n$ denotes
the maximal element of $\mathbb{U}$, while $1$ denotes the minimal
one). The Induction principle, we have just described, can be written
as

\[
%P\left({\bold 1}\right)\wedge \forall x\exists y\: \left(\left(x\, <\, y\right)\wedge %\left(P\left(x\right)\rightarrow P\left(y\right)\right)\right)\wedge \neg P\left({\bold n}\right).\]
P({\bold 1})\wedge \forall x\exists y\: \left( (x < y \wedge P(x)) \rightarrow P(y) \right) \wedge P({\bold n}). \]
The translation into propositional logic gives the following set of
clauses\begin{eqnarray*}
P_{1},\; \neg P_{n} &  & \\
\bigvee _{j=i+1}^{n} S_{i,j} &  & i \in [n-1]\\
\neg S_{i,j}\vee \neg P_{i}\vee P_{j} &  & i,j \in [n], i<j.
\end{eqnarray*}
The relativized version's translation is\begin{eqnarray}
P_{1},\; \neg P_{n} &  & \nonumber \\
R_{1}^{p},\: R_{n}^{p} &  & p \in [d] \nonumber \\
\bigvee _{j=i+1}^{n} S_{i,j} &  & i \in [n-1] \label{IndRSkolem}\\
\neg S_{i,j}\vee \neg \mathcal{R}_{i}\vee R_{j}^{p} &  & i,j \in [n], i<j,\: p \in [d] \label{IndRSR}\\
\neg S_{i,j}\vee \neg \mathcal{R}_{i}\vee \neg P_{i}\vee P_{j} &  &  i,j \in [n], i<j.\label{IndRSP}
\end{eqnarray}

\subsection{The upper bound}

\begin{proposition}
There is an $O\left(dn^{2}\right)$ size Res$^*\left(d+1\right)$ refutation of $d$-RIP$_n$.
\end{proposition}
\begin{proof}
We first apply $\wedge$-introduction between the clauses (\ref{IndRSR}) and (\ref{IndRSP}), $d$ times, to get the clauses\begin{eqnarray*}
\neg S_{i,j}\vee \neg \mathcal{R}_{i}\vee \neg P_{i}\vee \left(\mathcal{R}_{j}\wedge P_{j}\right) &  & 1\leq i<j\leq n.
\end{eqnarray*}
For every $i$ we resolve these with the clauses (\ref{IndRSkolem})
to get\begin{eqnarray}
\neg \mathcal{R}_{i}\vee \neg P_{i}\vee \bigvee _{j\geq i+1}\left(\mathcal{R}_{j}\wedge P_{j}\right). &  & \label{IndRk-th-stage-almost}
\end{eqnarray}

The $i$th stage clause is now\begin{eqnarray}
\bigvee _{j\geq i+1}\left(\mathcal{R}_{j}\wedge P_{j}\right). &  & \label{IndRk-th-stage}
\end{eqnarray}
For $i=1$ it is derived by resolving (\ref{IndRk-th-stage-almost})
with the pure-literal axioms $P_{1}$ and $R_{1}^{p}$, $p\in \left[d\right]$.
The induction step is pretty easy: we resolve the $(i-1)$th stage
clause $\bigvee _{j\geq i}\left(\mathcal{R}_{j}\wedge P_{j}\right)$ with
the clause (\ref{IndRk-th-stage-almost}) to get the $i$th stage
clause (\ref{IndRk-th-stage}). After the $(n-1)$th stage we have
derived the pure-term clause $\mathcal{R}_{n}\wedge P_{n}$. We now
weaken it to $P_{n}$ and then resolve it with the axiom $\neg P_{n}$
to get the desired empty clause.

The number of resolution steps is $O\left(dn^{2}\right)$. At each
stage we have resolved the clause obtained at the previous stage only
once, therefore the Res$\left(d+1\right)$ refutation, we have constructed,
 is tree-like.
\end{proof}

\subsection{The lower bound}

We will first prove it for $d=1$, \mbox{i.e.} that RIP$_n$
is exponentially hard for tree-like Resolution. We will then generalize
it to any $d$.

\begin{proposition} 
Any tree-like Resolution refutation of RIP$_n$
is of size $2^{\Omega \left(n\right)}$.
\end{proposition}
\begin{proof}
We will use an adversary strategy against
the decision tree solving the search problem.

We say that the variables $P_{i}$, $R_{i}$ and $S_{i,j}$ for $j>i$
are \emph{associated} to the $i$th element. When one of these has
been queried for the first time by Prover, Adversary fixes all of
them, so that the $i$th element becomes \emph{busy} (\mbox{cf.} definition of business in Section~\ref{MEP-Res}). Initially only
the minimal and maximal elements are busy as the singleton clauses $P_1$, $R_1$, $\neg P_{n}$
and $R_{n}$ force the values of the corresponding variables. For
technical reasons only, we assume that the $\left(n-1\right)$th
element is busy too, by setting $S_{n-1, n}=\top $, $R_{n-1}=\top $
and $P_{n-1}=\bot $. The elements that are not busy we call {}\emph{free},
with the exception of those below the \emph{source}. The source is the biggest
element $j$, such that $R_{j}=P_{j}=\top $. Initially the source
is the first element. It is important to note that no contradiction
can be found as far as there is at least one free element bigger than
the source. All the variables associated to the elements smaller than
the source are set (consistently with the axioms) in the current partial
assignment. Thus there are free elements only between the source and
the maximal element. Informally speaking, Prover's strategy is moving
the source towards the end of the universe, the $(n-2)$th element.

We will prove that at any stage in the Prover-Adversary game, the
number of free elements can be used to lower bound the subtree, rooted
at the current node of the tree. More precisely, if $T\left(m\right)$
is the size of the subtree rooted at a node, where there are $m$
such elements, we will show that $T\left(m\right)\geq \varphi _{m}$.
Here $\varphi _{m}$ is the $m$th Fibonacci number, defined by \[
\begin{array}{l}
 \varphi _{0}=\varphi _{1}=1\\
 \varphi _{m}=\varphi _{m-1}+\varphi _{m-2}\quad \textrm{for }m\geq 2.\end{array}
\]
 Initially, we have $n-3$ free elements bigger than the source, therefore
the inequality we claim, together with the known asymptotic $\varphi _{m}\sim \frac{1}{\sqrt{5}}\left(\frac{1+\sqrt{5}}{2}\right)^{m}$,
implies the desired lower bound. 

What remains is to prove $T\left(m\right)\geq \varphi _{m}$. We use
induction on $m$. The basis cases $m=0$ or $m=1$ are trivial. To
prove the induction step, we consider all the possibilities for a
Prover's query:

\begin{enumerate}
\item It is about either a busy element or an element below the source. As already explained, the value
of such a variable is already known in the current partial assignment.
Adversary answers; the value of $m$ does not change.
\item The query is about a free element $i$, and recall that it is bigger
than the source. If the variable queried is either $R_{i}$ or $P_{i}$,
Adversary first sets $S_{i,n}=\top $, $S_{i,j}=\bot $ for all $j$,
$i<j<n$, and then allows Prover a free choice between the two possibilities: either $R_{i}=\top $,
$P_{i}=\bot $ or $R_{i}=\bot $, $P_{i}=\top $. If the variable
queried is $S_{i,j}$, for some $j>i$, Adversary first sets $R_{i}=P_{i}=\bot $
$S_{i,l}=\top $ for all $l\neq j$, and then allows Prover a free choice between either $S_{i,j}=\top $
or $S_{i,j}=\bot $.
% In any of the above cases, Adversary was free
%to choose the value of the variable queried, i.e. to force a branch
%of Prover's decision tree, while keeping the current partial assignment
%consistent.
The number of free elements, $m$, decreases by one. Therefore
we have \[
T\left(m\right)\geq 2T\left(m-1\right).\]
 By the induction hypothesis $T\left(m-1\right)>\varphi _{m-1}$,
and then $T\left(m\right)\geq 2\varphi _{m-1}\geq \varphi _{m}$.
\item The query is about the source, \mbox{i.e.} the variable queried is $S_{i,j}$,
where $i$ is the source's index. If the $j$th element is busy,
Adversary answers $\bot $. If the $j$th element is free, but far
away from the source, that is there are at least two free elements
between the source and the $j$th element, Adversary answers $\bot $,
too. Neither the position of the source nor the value of $m$ changes.
The only remaining case is when the $j$th element is both free and
near to the source, that is one of the two smallest free elements,
bigger than the source. Adversary now offers Prover a free choice to move the source to
any of these two elements, by giving the corresponding answer: $\top $---source is moved to the $j$th element or $\bot $---source is moved
to the other nearest element. In one of these choices $m$ decreases
by one, and in the other it decreases by two. Therefore we have\[
T\left(m\right)\geq T\left(m-1\right)+T\left(m-2\right).\]
The induction hypothesis gives $T\left(m-1\right)\geq \varphi _{m-1}$
and $T\left(m-2\right)\geq \varphi _{m-2}$. Thus \[
T\left(m\right)\geq \varphi _{m-1}+\varphi _{m-2}=\varphi _{m}.\]
\end{enumerate}
This completes the proof.
\end{proof}
We will show how to modify the proof in order to prove a more general statement.
\begin{proposition}
 Any Res$^*\left(d\right)$ refutation of $d$-RIP$_n$ is of size $2^{\Omega \left(n/d\right)}$.
\end{proposition}
\begin{proof}
The proof is very similar to the previous one, so we will explain what changes should be made in there. Our task is now to prove that $T\left(m\right)\geq {\rho _{d}}^{m}$. Where $\rho _{d}$ is the largest real positive root of the equation
\[
x^{d+1}-x-1=0.\]
It is not hard to see that \[
1+\frac{\alpha }{d}\leq \rho _{d}\leq 1+\frac{\beta }{d}\]
for some appropriately chosen constants $\alpha $ and $\beta $.
Thus we would get the desired result as $T\left(m\right)\geq \left(1+\frac{\alpha }{d}\right)^{m}\geq e^{m/\left(1+d/\alpha \right)}$ (we have used the known inequality $\left(1+x\right)^{1+1/x}>e$ for $x>0$).

We first set the variables, associated to the last $d+1$ elements,
by setting $\mathcal{R}_{j}=\top $ (i.e. $R_{j}^{p}=\top $ for all
$p\in \left[d\right]$), $P_{j}=\bot $ and $S_{j,n}=\top $ for all
$j$, $n-d\leq j\leq n-1$.

Prover now can query $d$-disjunctions instead of single variables.
Adversary first simplifies the query, using the current partial assignment,
and then answers as follows:

\begin{enumerate}
\item The resulting query evaluates to either $\bot $ or $\top $ under
the current partial assignment. Adversary replies with the corresponding
value. Clearly the number of free elements, $m$, does not change. 
\item The resulting query involves only free elements; recall that they
all are bigger than the source. Suppose there are $d'$ ($d'\leq d$) such
elements. Adversary allows Prover a free choice between $\bot$ and $\top$,
without moving the source to one of the new elements. This is because
the only way to force such a movement is a positive answer to a query
$\mathcal{R}_{j}\wedge P_{j}$ which is of size $d+1$. The number
of free elements decreases by $d'$, and therefore we have \[
T\left(m\right)\geq 2T\left(m-d'\right).\]
By the induction hypothesis $T\left(m-d'\right)\geq {\rho _{d}}^{m-d'}\geq {\rho _{d}}^{m-d}$,
and then $T\left(m\right)\geq 2{\rho _{d}}^{m-d}>{\rho _{d}}^{m-d}+{\rho _{d}}^{m-d-1}={\rho _{d}}^{m}$.
\item The resulting query involves the source, \mbox{i.e.} contains variable(s)
of the form $S_{i,j}$, where $i$ is the source's index and $j$ is
a free element. Denote the set of all such elements by $\mathbb{J}$,
i.e. $\mathbb{J}=\left\{ j\mid S_{i,j}\textrm{ is in the query and }j\textrm{ is free}\right\} $.
If all the elements in $\mathbb{J}$ are far away from the source,
that is at distance at least $d+2$, Adversary first sets all the
$S_{i,j}$ $\bot $ and then answers the resulting query as in 
Case 2. In the other case, when at least one element from $\mathbb{J}$
is near to the source, it is always possible to move the source to two elements
between $1$st or $(d+1)$th nearest free element at worst, and the choice of which may be given to Prover. In the
former case we have an instance with at least $m-d$ free elements,
and in the latter with at least $m-d-1$. This gives\[
T\left(m\right)\geq T\left(m-d\right)+T\left(m-d-1\right).\]
The induction hypothesis gives $T\left(m-d\right)\geq {\rho _{d}}^{m-d}$
and $T\left(m-d-1\right)\geq {\varphi _{d}}^{m-d-1}$. Thus \[
T\left(m\right)\geq {\rho _{d}}^{m-d}+{\rho _{d}}^{m-d-1}={\rho _{d}}^{m}.\]
\end{enumerate}
This completes the proof.
\end{proof}

\section{Separating p-Res$(1)$ and p-Res$(2)$}
\label{sec:res}

We recall $\mathrm{RLNP}_n$ and its salient properties of being polynomial to refute in $\mathrm{Res}(2)$, but exponential in $\mathrm{Res}(1)$ (as in Section~\ref{MEP-Res}). Polynomiality clearly transfers to fpt-boundedness in $\mathrm{p\mbox{-}Res}(2)$, so we address the lower bound for $\mathrm{p\mbox{-}Res}(1)$.

\subsection{Lower bound: A strategy for Adversary over $\mathrm{RLNP}_n$}

We will give a strategy for Adversary in the game representation of a $\mathrm{p\mbox{-}Res}(1)$ refutation. The argument used in Section~\ref{MEP-Res} does not adapt to the parameterized case, so we instead use a technique developed specifically for the parameterized Pigeonhole principle in \cite{BGLRjournal}.

Recall that a \emph{parameterized clause} is of the form $\neg v_1 \vee \ldots \vee \neg v_{k+1}$ (where each $v_i$ is some $R$, $L$ or $S$ variable). The $i,j$ appearing in $R_i$, $L_{i,j}$ and $S_{i,j}$ are termed \emph{co-ordinates}. We define the following \emph{random restrictions}. Set $R_n:=\top$. Randomly choose $i_0 \in [n-1]$ and set $R_{i_0}:=\top$ and $L_{i_0,n}=S_{i_0,n}:=\top$. Randomly choose $n-\sqrt{n}$ elements from $[n-1]\setminus {i_0}$, and call this set $\mathbb{C}$. Set $R_i := \bot$ for $i \in \mathbb{C}$. Pick a random bijection $\pi$ on $\mathbb{C}$ and set $L_{i,j}$ and $S_{i,j}$, for $i,j \in \mathbb{C}$, according to whether $\pi(j)=i$. Set $L_{i,j}=L_{j,i}=S_{i,j}=S_{j,i}:=\bot$, if $j \in \mathbb{C}$ and $i \in [n] \setminus (\mathbb{C} \cup \{i_0\})$.

What is the probability that a parameterized clause is \textbf{not} evaluated to true by the random assignment? We allow that each of $\neg R_n$, $\neg R_{i,0}$, $\neg L_{i_o,n}$ and $\neg S_{i_0,n}$ appear in the clause---leaving $k+1-4=k-3$ literals, within which must appear $\sqrt{(k-3)/4}$ distinct co-ordinates. The probability that some $\neg R_i$, $i \notin \{i_0,n\}$, fails to be true is bound above by the probability that $i$ is in $[n-1] \setminus (\mathbb{C} \cup \{i_0\})$---which is $\leq \frac{\sqrt{n}-2}{n-2}\leq \frac{1}{\sqrt{n}}$. The probability that some $\neg L_{i,j}$ fails to be true, where one of the co-ordinates $i,j$ is possibly mentioned before and $(i,j) \neq (i_0,n)$, is bound above by the probability that both $i,j$ are in $[n] \setminus \mathbb{C}$ plus the probability that both $i,j$ are in $\mathbb{C}$ and $i = \pi(j)$. This gives the bound $\leq \frac{\sqrt{n}}{n} \cdot \frac{\sqrt{n}-1}{n-1} +  \frac{n-\sqrt{n}}{n} \cdot \frac{n-\sqrt{n}-1}{n-1} \cdot \frac{1}{n-\sqrt{n}-1} \leq \frac{2}{n}  \leq \frac{1}{\sqrt{n}}$. Likewise with $\neg S_{i,j}$.  Thus we get that the probability that a parameterized clause is \textbf{not} evaluated to true by the random assignment is $\leq \frac{1}{\sqrt{n}}^{\sqrt{(k-3)/4}}= n^{-\sqrt{(k-3)/16}}$.

Now we are ready to complete the proof. Suppose fewer than $n^{\sqrt{(k-3)/16}}$ parameterized clauses appear in a $\mathrm{p\mbox{-}Res}(1)$ refutation of $\mathrm{RLNP}_n$, then there is a random restriction as per the previous paragraph that evaluates all of these clauses to true. What remains is a $\mathrm{Res}(1)$ refutation of $\mathrm{RLNP}_{\sqrt{n}}$, which must be of size larger than $n^{\sqrt{(k-3)/16}}$ itself, for $n$ sufficiently large (see \cite{Rel-sep}). Thus we have proved.
\begin{theorem}
Every $\mathrm{p\mbox{-}Res}(1)$ refutation of $\mathrm{RLNP}_n$ is of size $\geq n^{\sqrt{(k-3)/16}}$.
\end{theorem}

\section{Separating p-Res$^*(1)$ and p-Res$^*(2)$}
\label{sec:tree-res}

Let us recall the important properties of $\mathrm{IP}_n$ and $\mathrm{RIP}_n$, from the perspective of Section~\ref{IP-TlRes}. $\mathrm{IP}_n$ admits refutation in $\mathrm{Res}^*(1)$ in polynomial size, as does $\mathrm{RIP}_n$ in $\mathrm{Res}^*(2)$. But all refutations of $\mathrm{RIP}_n$ in $\mathrm{Res}^*(1)$ are of exponential size. In the parameterized world things are not quite so well-behaved. Both $\mathrm{IP}_n$ and $\mathrm{RIP}_n$ admit refutations of size, say, $\leq 4 k!$ in $\mathrm{p\mbox{-}Res}^*(1)$; just evaluate variables $S_{i,j}$ from $i:=n-1$ downwards. Thus ask in sequence 
\[ S_{n-1,n}, S_{n-2,n-1}, S_{n-2,n}, \ldots \ldots, S_{n-k,n-k+1}, \ldots, S_{n-k,n}, \] 
each level $S_{n-i,n-i+1}, \ldots, S_{n-i,n}$ surely yielding a true answer. Clearly this is an fpt-bounded refutation. We are forced to consider something more elaborate, and thus we introduce the \emph{Relativized Vectorized Induction Principle} $\mathrm{RVIP}_n$ below. Roughly speaking, we stretch each single level of $\mathrm{RIP}_n$ into $n$ copies of itself in $\mathrm{RVIP}_n$, to make things easier for Adversary.
\[
\begin{array}{cl}
R_1, P_{1,1}, R_n, \neg P_{n,j} & j \in [n] \\
\bigvee_{l>i, m\in[n]} S_{i,j,l,m} & i,j \in [n] \\
\neg S_{i,j,l,m} \vee \neg R_{i} \vee \neg P_{i,j} \vee R_l & i \in [n-1], j,l,m \in [n] \\
\neg S_{i,j,l,m} \vee \neg R_{i} \vee \neg P_{i,j} \vee P_{l,m} & i \in [n-1], j,l,m \in [n] \\
\end{array}
\]
% \noindent One may assume there are variables $P_{1,j}$ for $j \in [n]$ that are unconstrained (or imagine these do not exist, according to taste.

\subsection{Lower bound: A strategy for Adversary over $\mathrm{RVIP}_n$}

We will give a strategy for Adversary in the game representation of a $\mathrm{Res}^*(1)$ refutation. For convenience, we will assume that Prover never questions the same variable twice (this saves us from having to demand trivial consistencies in future evaluations). 

Information conceded by Adversary of the form $R_i, \neg R_i, P_{i,j}$ and $S_{i,j,l,m}$ makes the element $i$ \emph{busy}. $\neg P_{i,j}$ and $\neg S_{i,j,l,m}$ do not make $i$ busy, in the case of $\neg P_{i,j}$ this is a departure from earlier definitions of business (due to the vectorization, there are now $n$ ways that some $i$ can become true as $P_{i,j}$). 
The \emph{source} is the largest element $i$ for which there is a $j$ such that Adversary has conceded $R_i \wedge P_{i,j}$. Initially, the source is $1$. Adversary always answers $R_1, P_{1,1},$ $R_n, \neg P_{n,j}$ (for  $j \in [n]$), according to the axioms. Thus $i:=1$ and $n$ are somehow special, and the size of the set inbetween is $n-2$. In the following, $i$ refers to the first index of a variable.

If $i$ is below the source. When Adversary is asked $R_i$, $P_{i,j}$ or $S_{i,j,l,m}$, then he answers $\bot$.

If $i$ is above the source. When Adversary is asked $R_i$, or $P_{i,j}$, then he gives Prover a free choice unless: 1.) $R_i$ is asked when some $P_{i,j}$ was previously answered $\top$ (in this case $R_i$ should be answered $\bot$); or 2.) Some $P_{i,j}$ is asked when $R_{i}$ was previously answered $\top$ (in this case $P_{i,j}$ should be answered $\bot$). When Adversary is asked $S_{i,j,l,m}$, then again he offers Prover a free choice. If Prover chooses $\top$ then Adversary sets $P_{i,j}$ and $R_i$ to $\bot$.

Suppose $i$ is the source. Then Adversary answers $P_{i,j}$ and $S_{i,j,l,m}$ as $\bot$, unless $R_i \wedge P_{i,j}$ witnesses the source. If $R_i \wedge P_{i,j}$ witnesses the source, then, if $l$ is not the next non-busy element above $i$, answer $S_{i,j,l,m}$ as $\bot$. If $l$ is the next non-busy element above $i$, then give $S_{i,j,l,m}$ a free choice, unless $\neg P_{l,m}$ is already conceded by Adversary, in which case answer $\bot$.  If Prover chooses $\top$ for $S_{i,j,l,m}$ then Adversary sets $R_l$ and $P_{l,m}$ to $\top$.

Using this strategy, Adversary can not be caught lying until either he has conceded that $k$ variables are true, or he has given Prover at least $n-2$ free choices.

Let $T(p,q)$ be some monotone decreasing function that bounds the size of the game tree from the point at which Prover has answered $p$ free choices $\top$ and $q$ free choices $\bot$. We can see that $T(p,q) \geq T(p+1,q) + T(p,q+1) +1$ and $T(k,n-2-k)\geq 0$. The following solution to this recurrence can be found in \cite{FOCS2007journal}.
\begin{corollary}
There is an $f \in \Omega(n^{k/16})$ \mbox{s.t.} every $\mathrm{p\mbox{-}Res}^*(1)$ refutation of $\mathrm{RVIP}_n$ is of size $\geq f(n)$.
\end{corollary}
We may increase the number of relativizing predicates to define $d$-$\mathrm{RVIP}_n$.
\[
\begin{array}{cl}
R^1_1,\ldots,R^d_1, P_{1,1}, R^1_n,\ldots,R^d_n \neg P_{n,j} & j \in [n] \\
\bigvee_{l>i, m\in[n]} S_{i,j,l,m} & i,j \in [n] \\
\neg S_{i,j,l,m} \vee \neg R^1_{i} \vee \ldots \vee \neg R^d_{i} \vee \neg P_{i,j} \vee R^1_l & i \in [n-1], j,l,m \in [n] \\
\vdots \\
\neg S_{i,j,l,m} \vee \neg R^1_{i} \vee \ldots \vee \neg R^d_{i} \vee \neg P_{i,j} \vee R^r_l & i \in [n-1], j,l,m \in [n], r \in [d] \\
\neg S_{i,j,l,m} \vee \neg R^1_{i} \vee \ldots \vee \neg R^d_{i} \vee \neg P_{i,j} \vee P_{l,m} & i \in [n-1], j,l,m \in [n] \\
\end{array}
\]
\noindent We sketch how to adapt the previous argument in order to demonstrate the following.
\begin{corollary}
There is an $f \in \Omega(n^{k/16d})$ \mbox{s.t.} every $\mathrm{p\mbox{-}Res}^*(d)$ refutation of $\mathrm{RVIP}^{d}_n$ is of size $\geq f(n)$.
\end{corollary}
\noindent We use essentially the same Adversary strategy in a branching $d$-program. We answer questions $\ell_1 \vee \ldots \vee \ell_d$ as either forced or free exactly according to the disjunction of how we would have answered the corresponding $\ell_i$s, $i \in [d]$, before. That is, if one $\ell_i$ would give Prover a free choice, then the whole disjunction is given as a free choice.
The key point is that once some disjunction involving some subset of $R^1_i,\ldots,R^d_i$ or $P_{i,j}$ (never all of these together, of course), is questioned then, on a positive answer to this, the remaining unquestioned variables of this form should be set to $\bot$. This latter rule introduces the factor of $d$ in the exponent of $n^{k/16d}$.

\subsection{Upper bound: a $\mathrm{Res}^*(d+1)$ refutation of $\mathrm{RVIP}^d_n$}

We encourage the reader to have a brief look at the simpler, but very similar, refutation of $\mathrm{RIP}_n$ in $\mathrm{Res}^*(2)$, of size $O(n^2)$, as depicted in Figure~\ref{fig:figure1}.
\begin{figure}
\[
\xymatrix{
\neg R_n \vee \neg P_n ? \ar[d]_{\top} \ar[r]^{\bot} & \# & \\
\neg R_{n-1} \vee \neg P_{n-1} ? \ar[d]_{\top} \ar[r]^{\bot} & S_{n-1,n} ? \ar[d]_{\top} \ar[r]^{\bot} & \# \\
\vdots \ar[d]_\top & \# & \\
\neg R_1 \vee \neg P_1 ? \ar[d]_{\top} \ar[r]^{\bot} & S_{1,n} ? \ar[d]_{\top} \ar[r]^{\bot} & \cdots \ar[r]^{\bot} & S_{1,2} ? \ar[d]_{\top} \ar[r]^{\bot} & \# \\
\# & \# & & \# \\
}
\]
\caption{Refutation of $\mathrm{RIP}_n$ in $\mathrm{Res}^*(2)$.}
\label{fig:figure1}
\end{figure}
\begin{proposition}
There is a refutation of $d$-$\mathrm{RVIP}_n$ in $\mathrm{Res}^*(d+1)$, of size $O(n^{d+4})$.
\end{proposition}
\begin{proof}
We give the branching program for $d:=1$ in Figure~\ref{fig:figure2}. The generalization to higher $d$ is clear: substitute questions of the form $\neg R_i \vee \neg P_{i,j}$ by questions of the $\neg R^1_i \vee \ldots \vee R^d_i \vee \neg P_{i,j}$.
\end{proof}
\begin{figure}
\[
\xymatrix{
\neg R_n \vee \neg P_{n,n} ? \ar[d]_{\top} \ar[r]^{\bot} & \# & \\
\vdots \ar[d]_{\top} \\
\neg R_n \vee \neg P_{n,1} ? \ar[d]_{\top} \ar[r]^{\bot} & \# & \\
\neg R_{n-1} \vee \neg P_{n-1,n} ? \ar[d]_{\top} \ar[r]^{\bot} & S_{n-1,n,n,n} ? \ar[d]_{\top} \ar[r]^{\bot} & \cdots \ar[r]^{\bot} &  S_{n-1,n,n,1} ? \ar[d]_{\top} \ar[r]^{\bot} & \# \\
\vdots \ar[d]_\top & \# & & \# & \\
\neg R_{n-1} \vee \neg P_{n-1,1} ? \ar[d]_{\top} \ar[r]^{\bot} & S_{n-1,1,n,n} ? \ar[d]_{\top} \ar[r]^{\bot} & \cdots \ar[r]^{\bot} &  S_{n-1,1,n,1} ? \ar[d]_{\top} \ar[r]^{\bot} & \# \\
\vdots \ar[d]_\top & \# & & \# & \\
\vdots \ar[d]_\top  \\
\neg R_1 \vee \neg P_{1,n} ? \ar[d]_{\top} \ar[r]^{\bot} & S_{1,n,n,n} ? \ar[d]_{\top} \ar[r]^{\bot} & \cdots \ar[r]^{\bot} & \cdots \ar[r]^{\bot} & S_{1,n,2,1} ? \ar[d]_{\top} \ar[r]^{\bot} & \# \\
\vdots \ar[d]_\top & \# & & & \# \\
\neg R_1 \vee \neg P_{1,1} ? \ar[d]_{\top} \ar[r]^{\bot} & S_{1,1,n,n} ? \ar[d]_{\top} \ar[r]^{\bot} & \cdots \ar[r]^{\bot} & \cdots \ar[r]^{\bot} & S_{1,1,2,1} ? \ar[d]_{\top} \ar[r]^{\bot} & \# \\
\# & \# & & & \# \\
}
\]
\caption{Refutation of $\mathrm{RVIP}_n$ in $\mathrm{Res}^*(2)$.}
\label{fig:figure2}
\end{figure}

\section{Final remarks}

It is most natural when looking for separators of $\mathrm{p\mbox{-}Res}^*(1)$ and $\mathrm{p\mbox{-}Res}^*(2)$ to look for CNFs, like $\mathrm{RVIP}_n$ that we have given. $\mathrm{p\mbox{-}Res}^*(2)$ is naturally able to process $2$-DNFs and we may consider $\mathrm{p\mbox{-}Res}^*(1)$ acting on $2$-DNFs, when we think of it using any of the clauses obtained from those $2$-DNFs by distributivity. In this manner, we offer the following principle as being fpt-bounded for $\mathrm{p\mbox{-}Res}^*(2)$ but not fpt-bounded for $\mathrm{p\mbox{-}Res}^*(1)$. Consider the two axioms $\forall x (\exists y \ \neg S(x,y) \wedge T(x,y)) \vee P(x)$ and $\forall x,y\ T(x,y) \rightarrow S(x,y)$. This generates the following system $\Sigma_{PST}$ of $2$-DNFs.
\[
\begin{array}{cl}
P_i \vee \bigvee_{j \in [n]} (\neg S_{i,j} \wedge T_{i,j}) & i \in [n] \\
\neg T_{i,j} \vee S_{i,j} & i,j \in [n]
\end{array}
\]
\noindent Note that the expansion of $\Sigma_{PST}$ to CNF makes it exponentially larger. It is not hard to see that $\Sigma_{PST}$ has refutations in $\mathrm{p\mbox{-}Res}^*(2)$ of size $O(kn)$, while any refutation in $\mathrm{p\mbox{-}Res}^*(1)$ will be of size $\geq n^{k/2}$.

All of our upper bounds, \mbox{i.e.} for both $\mathrm{RVIP}_n$ and $\mathrm{RLNP}_n$, are in fact polynomial, and do not depend on $k$. That is, they are fpt-bounded in a trivial sense. If we want examples that depend also on $k$ then we may enforce this easily enough, as follows. For a set of clauses $\Sigma$, build a set of clauses $\Sigma'_k$ with new propositional variables $A$ and $B_1,B'_1,\ldots,B_{k+1},B'_{k+1}$. From each clause $\mathcal{C} \in \Sigma$, generate the clause $A \vee \mathcal{C}$ in $\Sigma'_k$. Finally, augment $\Sigma'_k$ with the following clauses: $\neg A \vee B_1 \vee B'_1$, \ldots, $\neg A \vee B_{k+1} \vee B'_{k+1}$. If $\Sigma$ admits refutation of size $\Theta(n^c)$ in $\mathrm{p\mbox{-}Res}^*(d)$ then $(\Sigma'_k,k)$ admits refutation of size $\Theta(n^c+2^{k+1})$. The parameterized contradictions so obtained are no longer ``strong'', but we could even enforce this by augmenting instead a Pigeonhole principle from $k+1$ to $k$.

It seems hard to prove p-Res$(1)$ lower bounds for parameterized
$k$-clique on a random graph \cite{BGL-SAT}, but we now introduce a contradiction that looks similar but for which lower bounds should be easier.
It is a variant of the Pigeonhole principle which could give us another very natural separation of $\mathrm{p\mbox{-}Res}(1)$ from $\mathrm{p\mbox{-}Res}(2)$. Define the contradiction PHP$_{k+1,n,k}$, on variables $P_{i,j}$ ($i \in [k+1]$ and $j \in [n]$) and $Q_{i,j}$ ($i \in [n]$ and $j \in [k]$), and with clauses:
\[
\begin{array}{ll}
\neg P_{i,j} \vee \neg P_{l,j} & i \neq l \in [k+1]; j \in [n] \\
\neg Q_{i,j} \vee \neg Q_{l,j} & i \neq l \in [n]; j \in [k] \\
\bigvee_{j \in [n]} P_{i,j} & i \in [k] \\
\neg P_{i,j} \vee \bigvee_{l \in [k]} Q_{j,l} & j \in [n] \\
\end{array}
\]
\noindent We conjecture that this principle, which admits fpt-bounded refutation in $\mathrm{p\mbox{-}Res}(2)$, does not in $\mathrm{p\mbox{-}Res}(1)$.

We have left open the technical question as to whether suitably defined, further-relativized versions of RLNP$_n$ can separate $\mathrm{p\mbox{-}Res}(d)$ from $\mathrm{p\mbox{-}Res}(d+1)$. We conjecture that they can.

Finally, it is possible that the results of Section~\ref{sec:tree-res} might be derived in a simpler manner using the assymetric Prover-Delayer game of \cite{Beyersdorff13IPL}.

\section*{Acknowledgements}

The authors would like to thank Jan Kraj\'{\i}\v{c}ek for asking the question which led to writing this paper. Many thanks to S{\o}ren Riis for the helpful discussions, the first author has had with him. Finally, many thanks to our anonymous referees for pointing out important errors.
% Finally, the authors are grateful for the intercession of \mbox{St.} Peter and \mbox{St.} Paul.

\bibliographystyle{plain}
%\bibliography{ProofComplexity}

\end{document}